\documentclass[10pt,conference]{IEEEtran}\IEEEoverridecommandlockouts
\usepackage{epsfig,graphics,subfigure,psfrag,amsmath,cases}
\usepackage{latexsym,amssymb,amsmath,epsfig,subfigure,algorithm,mathtools}
\usepackage{algorithmic}
\usepackage{times}
\newtheorem{Thm}{Theorem}

\DeclareMathOperator{\Tr}{\mathrm{Tr}}
\DeclareMathOperator{\zero}{\mathbf{0}}
\DeclareMathOperator{\Rank}{\mathrm{Rank}}

\DeclareMathOperator{\diag}{\mathrm{diag}}

\DeclareMathOperator{\maxo}{maximize}

\newcommand{\qed}{\hfill \ensuremath{\blacksquare}}

\newcommand{\abs}[1]{\lvert#1\rvert}

\author{Marco Breiling\authorrefmark{2}, Derrick Wing Kwan Ng\authorrefmark{1}, Christian Rohde\authorrefmark{2}, Frank Burkhardt\authorrefmark{2}, and Robert Schober\authorrefmark{1}\thanks{\authorrefmark{1}The authors are  with the Institute for Digital Communications, Friedrich-Alexander-University Erlangen-N\"urnberg (FAU), Germany. \authorrefmark{2}The authors are with Fraunhofer Institute for Integrated Circuits (IIS), Germany. This work was supported in part by the AvH Professorship Program of the Alexander von Humboldt Foundation.

}}
\title{Resource Allocation for Outdoor-to-Indoor Multicarrier Transmission with Shared UE-side Distributed Antenna Systems }
\begin{document}
\maketitle
\begin{abstract}
In this paper, we study the resource allocation algorithm design  for  downlink
multicarrier transmission with a shared user equipment (UE)-side distributed antenna system (SUDAS) which
 utilizes both licensed and unlicensed frequency bands  for improving the system throughput.  The joint UE selection and  transceiver processing matrix design is formulated as a non-convex optimization problem for the maximization of  the end-to-end system throughput (bits/s).  In order to obtain a tractable resource allocation algorithm, we first show that the optimal transmitter precoding and receiver post-processing matrices  jointly diagonalize the end-to-end communication channel. Subsequently,  the optimization problem is converted to a scalar optimization problem for multiple parallel channels, which is solved by using an asymptotically optimal iterative algorithm.
 Simulation results illustrate that the proposed resource allocation algorithm for the SUDAS
 achieves an excellent system performance and provides a  spatial multiplexing gain for single-antenna UEs.

\end{abstract}

\renewcommand{\baselinestretch}{0.93}
\large\normalsize

\section{Introduction}
\label{sect1}
Ubiquitous and high data rates   are a basic requirement for the next generation  wireless communication systems.
As a result, orthogonal frequency division multiple access (OFDMA) has been adopted as an air interface for high speed wideband communication systems,  due to its flexibility in resource allocation and resistances against multipath fading \cite{book:david_wirelss_com}.  On the other hand, multiple-input multiple-output (MIMO)  technology has received considerable interest in the past decades as it provides extra degrees of freedom in the spatial domain which facilitates a trade-off between multiplexing gain and diversity gain. Besides, distributed antenna systems  (DAS), a special form of MIMO,  can be deployed  to cover the dead spots in wireless networks, extending service coverage, improving spectral efficiency, and  mitigating interference. However, the number of antennas available at the user equipments (UEs) is constrained by the physical size of the devices in practice which leads to a limited spatial multiplexing gain in MIMO systems.
As an alternative, multiuser MIMO has been proposed to realize the potential performance gain of MIMO systems by sharing the antennas across the different terminals  of a communication system \cite{CR:virtual_MIMO0,JR:virtual_MIMO3}.    In \cite{CR:virtual_MIMO0}, the energy efficiency
 of a three-node multiuser MIMO system  was studied for the
compress-and-forward protocol. In \cite{JR:virtual_MIMO3}, optimal power allocation was investigated for the maximization of
the effective capacity of a multiuser MIMO system for the case when the receivers collaborate with each other. Recently, there has been a growing interest in combining the
concepts of multiuser MIMO, DAS, and OFDMA to improve the performance of wireless communication systems.  In \cite{JR:virtual_MIMO1}, the authors studied suboptimal resource allocation algorithm design for multiuser MIMO-OFDMA systems. In \cite{JR:virtual_MIMO2}, a utility-based low complexity scheduling scheme was proposed for multiuser MIMO-OFDMA systems to strike a balance between system throughput  and computational complexity.  On the other hand, the performance of multiuser MIMO in DAS with limited feedback was investigated in \cite{CR:virtual_MIMO}. However,  the system performance of the systems \cite{JR:virtual_MIMO1}--\cite{CR:virtual_MIMO}  is limited by the system bandwidth which is a very scarce resource in   the licensed frequency bands.  On the contrary, the unlicensed frequency spectrum around $60$ GHz offers a bandwidth of $7$ GHz for wireless communications. The utilization of both  licensed  and unlicensed  frequency bands introduces a paradigm shift
in system and resource allocation algorithm design due to
the resulting new opportunities and challenges. Yet, the potential performance gains of such a hybrid system have not been investigated in the literature.

 In this paper, we propose a  shared user equipment (UE)-side distributed antenna system (SUDAS) to assist the downlink communication.  In particular, SUDAS utilizes both  licensed  and unlicensed  frequency bands simultaneously to facilitate a spatial multiplexing gain for single-antenna receivers. We formulate the resource allocation algorithm design for  SUDAS assisted OFDMA downlink transmission systems as a non-convex optimization problem.  By exploiting the structure of the optimal base station (BS)
precoding and the SUDAS post-processing matrices, the considered
matrix optimization problem is transformed into an optimization problem with scalar optimization variables. Capitalizing on this transformation, we develop an iterative algorithm
which achieves the asymptotically optimal performance of the proposed SUDAS.

 The rest of the paper is organized as follows. In Section
\ref{sect:OFDMA_AF_network_model}, we outline the model for the
considered SUDAS  assisted OFDMA downlink transmission system. In Section \ref{sect:cross-layer system}, we formulate the
resource allocation as a non-convex optimization problem.  Simulation  results for the performance of the proposed algorithm are presented in Section \ref{sect:result-discussion}. In
Section \ref{sect:conclusion}, we conclude with a brief summary of
our results.

\section{SUDAS Assisted OFDMA Network Model}\label{sect:OFDMA_AF_network_model}
In this section, after introducing the notation used in this
paper, we present the adopted channel and signal models.

\subsection{Notation}
We use boldface capital and lower case letters to denote matrices and vectors, respectively. $\mathbf{A}^H$, $\det(\mathbf{A})$, $\Tr(\mathbf{A})$, and $\Rank(\mathbf{A})$ represent the  Hermitian transpose, determinant, trace, and rank of  matrix $\mathbf{A}$; $\mathbf{A}\succeq \mathbf{0}$ indicates that $\mathbf{A}$ is a   positive semidefinite matrix; $\mathbf{I}_N$ is the $N\times N$ identity matrix; $\mathbb{C}^{N\times M}$ denotes the set of all $N\times M$ matrices with complex entries; $\mathbb{H}^N$ denotes the set of all $N\times N$ Hermitian matrices; $\diag(x_1, \cdots, x_K)$ denotes a diagonal matrix with the diagonal elements given by $\{x_1, \cdots, x_K\}$; the circularly symmetric complex Gaussian (CSCG) distribution with mean  $\mu$ and variance $\sigma^2$ is denoted by ${\cal CN}(\mu,\sigma^2)$; $\sim$ stands for ``distributed as"; $[x]^+$ returns  $0$ when $x<0$ and returns $x$ if $x\ge 0$; $\cal E\{\cdot\}$ denotes
statistical expectation.
 \begin{figure}[t]
 \centering
\includegraphics[width=3.5in]{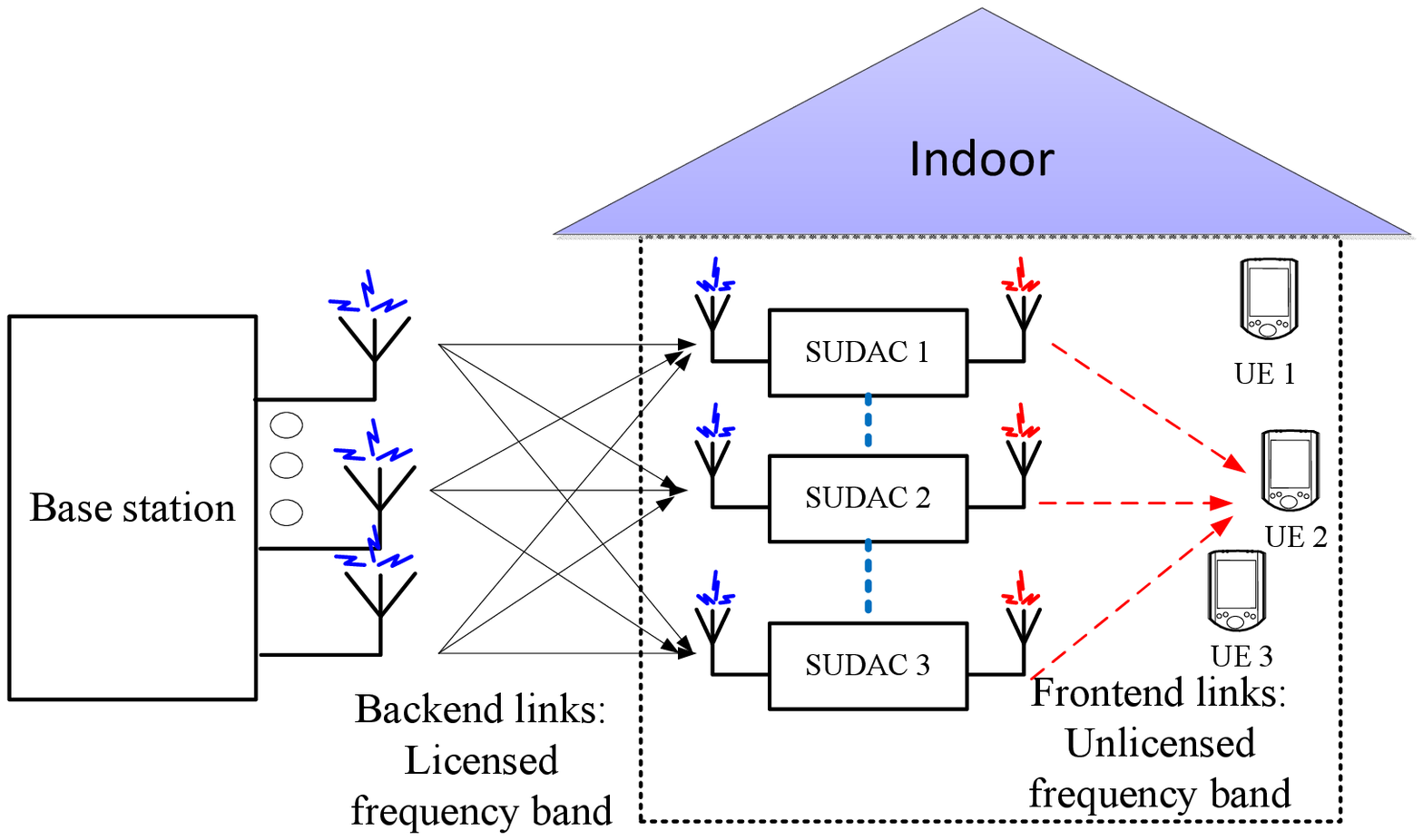}
 \caption{Downlink  communication system model with a base station (BS), $K=3$  user equipments (UEs), and $M=3$  SUDACs. The backend links use a licensed frequency band and the frontend links use an unlicensed frequency band such as the millimeter wave band (e.g. $\sim 60 $ GHz). }
 \label{fig:system_model}
\end{figure}

\subsection{SUDAS Downlink Communication Model}
We consider a SUDAS assisted
OFDMA downlink  transmission network which consists of one
BS, a SUDAS, and $K$ UEs, cf. Figure \ref{fig:system_model}. A SUDAS consists of $M$ shared user equipment (UE)-side distributed antenna components (SUDACs). A SUDAC is a small and cheap device\footnote{In practice, a SUDAC could be integrated into electrical devices such as electrical wall outlets, switches, and light outlets.} which utilizes both a  licensed and an unlicensed frequency band simultaneously for increasing the end-to-end communication data rate. Conceptually, a basic SUDAC is equipped with one antenna for use in a licensed band and one  antenna for use in an unlicensed band. Besides, a SUDAC is equipped with a mixer  to perform frequency up-conversion/down-conversion. Specifically, for the downlink, the SUDAC first receives the signal from the BS in a licensed frequency band (backend link).  Then the SUDAC processes the received signal and forwards the signal to the UEs in an unlicensed frequency band (frontend link). We note that since  signal reception and transmission at each SUDAC are separated in frequency, cf. Figure \ref{fig:system_model2}, simultaneous signal reception and transmission  can be performed  which is not possible for traditional relaying systems\footnote{Since the BS-to-SUDAS and SUDAS-to-UE links operate on two different frequencies, the SUDAS should not be considered  a relaying system \cite{JR:Jeff_7_ways}.} due to the limited spectrum availability in the licensed bands. In practice, a huge bandwidth is available in the unlicensed bands. For instance,  there is  nearly $7$ GHz of unlicensed frequency spectrum available  for information transmission in the  $57 - 64$ GHz band (millimeter wave band).
 \begin{figure}[t]
 \centering
\includegraphics[width=3.5in]{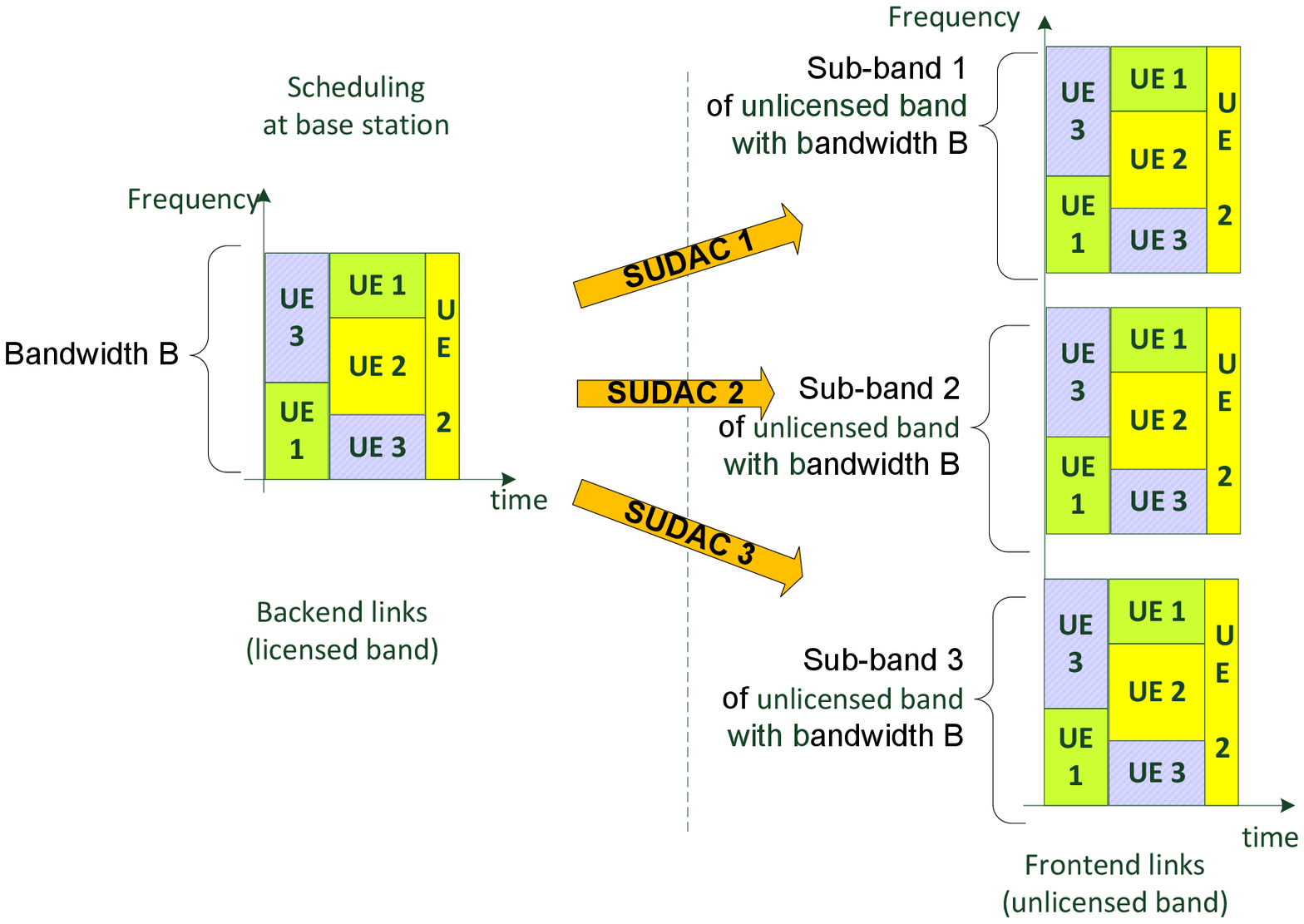}
 \caption{Illustration of signal forwarding from the licensed band to different unlicensed frequency sub-bands in the SUDAS. }
 \label{fig:system_model2}
\end{figure}
In this paper, we study the potential system performance gain for outdoor-to-indoor transmission achieved by the SUDAS. We  assume that the SUDACs are installed in electrical wall outlets and can cooperate with each other by sharing the channel state information and received signals, i.e., $\mathbf{y}^{[i,k]}_{\mathrm{S}}$, via low data-rate power line communication links. In other words, joint processing between SUDACs is possible such that the SUDACs can fully utilize their antennas\footnote{We note that different SUDAS configurations are possible in practice including non-cooperative SUDACs. In this paper,   we are interested in the case when the SUDACs are willing to cooperate to explore the maximum system performance.}.   Besides, the UEs only listen to the unlicensed frequency band.

\subsection{SUDAS Downlink Communication Channel Model}\label{sect:channel_model}
 \label{sect:channel_model}
 The BS is equipped with  $N_{\mathrm{T}}$ transmit antennas transmitting signals in a licensed frequency band.   The UEs are  single-antenna devices receiving the signal in the unlicensed frequency band. We focus on a wideband multicarrier communication system with $n_{\mathrm{F}}$ subcarriers. The  communication channel is  time-invariant within a scheduling slot. The BS performs spatial multiplexing in the licensed band.  The data symbol vector $\mathbf{d}^{[i,k]}\in
\mathbb{C}^{N_{\mathrm{S}}\times 1}$ on subcarrier $i\in\{1,\,\ldots,\,n_F\}$
 for UE $k$ is precoded at the BS as
\begin{eqnarray}
\mathbf{x}^{[i,k]}=\mathbf{P}^{[i,k]}\mathbf{d}^{[i,k]},\label{eqn:source_precoding}
\end{eqnarray}
where $\mathbf{P}^{[i,k]}\in\mathbb{C}^{N_{\mathrm{T}}\times N_{\mathrm{S}}}$ is the
 precoding matrix adopted by the BS on subcarrier $i$.  The signals received on subcarrier $i$ at the $M$ SUDACs  for UE $k$ are given by
\begin{eqnarray}
\mathbf{y}^{[i,k]}_{\mathrm{S}}&=&\mathbf{H}^{[i]}_{\mathrm{B}\rightarrow \mathrm{S}}\mathbf{x}^{[i,k]}+\mathbf{z}^{[i,k]},
\label{eqn:relay_channel_model:AF}
\end{eqnarray}
where $\mathbf{y}^{[i,k]}_{\mathrm{S}}=[{y}^{[i,k]}_{\mathrm{S}_1}\ldots {y}^{[i,k]}_{\mathrm{S}_m}\ldots {y}^{[i,k]}_{\mathrm{S}_M}]^T$ and ${y}^{[i,k]}_{\mathrm{S}_m}\in\{1,\ldots,M\}$ denotes the received signal at SUDAC $m$. $\mathbf{H}^{[i]}_{\mathrm{B}\rightarrow S} $ is the $ M\times  N_{\mathrm{T}}$ MIMO channel
matrix between the BS and the $M$ SUDACs on subcarrier $i$ and captures
the joint effects of  path loss, shadowing, and multi-path fading.  $\mathbf{z}^{[i,k]}$ is the
additive white Gaussian noise (AWGN) vector with distribution ${\cal
CN}(\zero,\mathbf{\Sigma})$ on subcarrier $i$ impairing the $M$ SUDACs where
$\mathbf{\Sigma}$ is an $M\times M$  diagonal
covariance matrix with each main diagonal element given by $N_0$.

Then,  each SUDAC performs frequency repetition in the unlicensed band. In particular, the  $M$ SUDACs   multiply  the received signal
vector on subcarrier $i$ by $\mathbf{F}^{[i,k]}\in\mathbb{C}^{M\times M}$ and forward
the processed signal vector to UE $k$ on subcarrier  $i$ in $M$  different independent frequency sub-bands in the unlicensed spectrum, cf. Figure \ref{fig:system_model2}.  In other words, each SUDAC forwards its received signal in a different sub-band and thereby avoids further multiple access interference in the unlicensed spectrum.

Then, the
signal  received  at UE $k$ on subcarrier $i$ from the SUDACs in the $M$ frequency bands, $\mathbf{y}^{[i,k]}_{\mathrm{S}\rightarrow\mathrm{UE}}\in\mathbb{C}^{M\times 1}$ ,  can be expressed as
\begin{eqnarray}
\hspace*{-3mm}&&\mathbf{y}^{[i,k]}_{\mathrm{S}\rightarrow\mathrm{UE}}\\
\hspace*{-3mm}&=&\notag\hspace*{-1mm}\mathbf{H}^{[i,k]}_{\mathrm{S}\rightarrow\mathrm{UE}}
\mathbf{F}^{[i,k]}\Big(\mathbf{H}^{[i]}_{\mathrm{B}\rightarrow\mathrm{S}}\mathbf{x}^{[i,k]}+\mathbf{z}^{[i]}\Big)\hspace*{-1mm}+\hspace*{-1mm}\mathbf{n}^{[i,k]}
\\
\hspace*{-3mm}&=&\hspace*{-1mm}\underbrace{\mathbf{H}^{[i,k]}_{\mathrm{S}\rightarrow\mathrm{UE}}\mathbf{F}^{[i,k]}\mathbf{H}^{[i]}_{\mathrm{B}\rightarrow\mathrm{S}}\mathbf{P}^{[i,k]}\mathbf{s}^{[i,k]}}_{
\mbox{desired
signal}}+\underbrace{\mathbf{H}^{[i,k]}_{\mathrm{S}\rightarrow\mathrm{UE}}\mathbf{F}^{[i,k]}\mathbf{z}^{[i]}}_{\mbox{amplified
noise}}\hspace*{-1mm}+\mathbf{n}^{[i,k]}.\notag
\end{eqnarray}
The $m$-th element of vector $\mathbf{y}^{[i,k]}_{\mathrm{S}\rightarrow\mathrm{UE}}$ represents the received signal at UE $k$ in the $m$-th unlicensed frequency sub-band. Besides, since the SUDACs forward the received signals in different frequency bands, $\mathbf{H}^{[i,k]}_{\mathrm{S}\rightarrow\mathrm{UE}}$  is a diagonal matrix with the diagonal elements representing the channel gain between the SUDACs and UE $k$ on  subcarrier $i$ in the unlicensed sub-band $m$. $\mathbf{n}^{[i,k]}\in\mathbb{C}^{M\times 1}$ is the AWGN
vector at UE $k$ on subcarrier $i$  with distribution ${\cal
CN}(\zero,\mathbf{\Sigma}_k)$.  $\mathbf{\Sigma}_k$
 is an ${M \times M}$ diagonal matrix and each main
diagonal element is equal to $N_0$. In order to simplify the subsequent
mathematical expressions and without loss of generality, we adopt
in the following a normalized noise variance of $N_0=1$ at all
receive antennas of the SUDACs and the UEs.

We assume that $ M \ge N_{\mathrm{S}}$ and  UE $k$ employs a linear receiver for estimating the  data vector symbol received in the $M$ different frequency bands in the unlicensed band.
The estimated data vector symbols, $\mathbf{\hat{d}}^{[i,k]}\in\mathbb{C}^{N_\mathrm{S}\times 1}$, on subcarrier
$i$ is given by:
\begin{eqnarray}
\mathbf{\hat{d}}^{[i,k]}=(\mathbf{W}^{[i,k]})^H\mathbf{y}^{[i,k]}_{\mathrm{S}\rightarrow\mathrm{UE}},
\end{eqnarray}
where $\mathbf{W}^{[i,k]}\in \mathbb{C}^{M\times M}$ is a
post-processing matrix used for subcarrier $i$ at UE $k$. Without loss of generality, we assume that  ${\cal
E}\{\mathbf{{d}}^{[i,k]}(\mathbf{{d}}^{[i,k]})^H\}=\mathbf{I}_{N_\mathrm{S}}$. As a result,
the mean square error (MSE) matrix for the data transmission on
subcarrier $i$ for UE $k$ via the SUDAS and the optimal post processing matric are given by
\begin{eqnarray}
\mathbf{E}^{[i,k]}\hspace*{-2mm}&=&\hspace*{-2mm}{\cal E}\{(\mathbf{\hat{d}}^{[i,k]}-\mathbf{{d}}^{[i,k]})(\mathbf{\hat{d}}^{[i,k]}-\mathbf{{d}}^{[i,k]})^H\}\notag\\
\hspace*{-2mm}&=&\hspace*{-2mm}\Big[\mathbf{I}_{N_\mathrm{S}}+(\mathbf{\Gamma}^{[i,k]})^H(\mathbf{\Theta}^{[i,k]})^{-1}
(\mathbf{\Gamma}^{[i,k]})\Big]^{-1},\\
\mbox{and }\mathbf{W}^{[i,k]}\hspace*{-2mm}&=&\hspace*{-2mm}(\mathbf{\Gamma}^{[i,k]}(\mathbf{\Gamma}^{[i,k]})^H
+\mathbf{\Theta}^{[i,k]})^{-1}\mathbf{\Gamma}^{[i,k]},
\label{eqn:AF-FD-self-interference}
\end{eqnarray}
respectively, where $\mathbf{\Gamma}^{[i,k]}$ is the effective MIMO channel
matrix between the BS and UE $k$ via the SUDAS on subcarrier $i$, and $\mathbf{\Theta}^{[i,k]}$
is the corresponding equivalent noise covariance matrix. These
matrices are given by
\begin{eqnarray}
\mathbf{\Gamma}^{[i,k]}\hspace*{-2mm}&=&\hspace*{-2mm}\mathbf{H}^{[i,k]}_{\mathrm{S}\rightarrow\mathrm{UE}}\mathbf{F}^{[i,k]}\mathbf{H}^{[i]}_{\mathrm{B}\rightarrow\mathrm{S}}\mathbf{P}^{[i,k]}
\quad \mbox{and}\quad\notag\\
\label{eqn:AF-FD-equivalent_noise}\mathbf{\Theta}^{[i,k]}\hspace*{-2mm}&=&\hspace*{-2mm}
\Big(\mathbf{H}^{[i,k]}_{\mathrm{S}\rightarrow\mathrm{UE}}\mathbf{F}^{[i,k]}\Big)
\Big(\mathbf{H}^{[i,k]}_{\mathrm{S}\rightarrow\mathrm{UE}}\mathbf{F}^{[i,k]}\Big)^H+\mathbf{I}_{M}.
\end{eqnarray}

\section{Resource Allocation and Scheduling Design }\label{sect:cross-layer system}
In this section, we first introduce the adopted system performance
measure. Then, the resource allocation and scheduling design is formulated as an optimization problem.
\subsection{System Throughput}
\label{subsect:Instaneous_Mutual_information}
The end-to-end achievable data rate $R^{[i,k]}$ on subcarrier $i$ between the BS and UE $k$ via the SUDAS
  is given by  \cite{JR:DP_diagonalization}
  \begin{eqnarray}\label{eqn:cap_log_det}
R^{[i,k]}=-\log_2\Big(\det[\mathbf{E}^{[i,k]}]\Big).
\end{eqnarray}
The data rate (bit/s)
delivered to UE $k$  can be expressed as
\begin{eqnarray}
\label{eqn:user_TP}{\cal \rho}^{[k]}=
\sum_{i=1}^{n_F}s^{[i,k]}R^{[i,k]},
\end{eqnarray}
where  $s^{[i,k]}\in\{0,1\}$ is the binary subcarrier allocation
indicator. The  weighted system throughput via the SUDAS  is given by
 \begin{eqnarray}
\label{eqn:avg-sys-TP} \mathrm{TP}({\cal P},{\cal
S})=\sum_{k=1}^K w^{[k]}{\cal \rho}^{[k]} ,
\end{eqnarray}
where ${\cal
P}=\{\mathbf{P}^{[i,k]},\mathbf{F}^{[i,k]}\}$ and
${\cal{S}}=\{s^{[i,k]}\}$ are the precoding and subcarrier allocation
policies, respectively.  $w^{[k]}$ is a positive constant which indicates the priority of different UEs. It is known that by adjusting the values of  $w^{[k]}$, different kinds of fairness such as max-min fairness and proportional fairness can be achieved \cite{JR:Cross_layer_weighted_sum_cap1,CN:Cross_layer_weighted_sum_cap1}.

\subsection{Problem Formulation}\label{sect:cross-Layer_formulation}
 The optimal precoding matrices, ${\cal P}^*=\{\mathbf{P}^{[i,k]*},\mathbf{F}^{[i,k]*} \}$, and the optimal subcarrier allocation policy, ${\cal
S}^*=\{s^{[i,k]*}\}$, can be obtained by solving the following optimization problem:
\begin{eqnarray}\label{eqn:cross-layer-formulation}
&&\hspace*{-3mm}\underset{{{\cal P},{\cal S}}}\maxo\ \
\mathrm{TP}({\cal P},{\cal S}) \notag\\
\mbox{s.t.} &\mbox{C1:}& \sum_{k=1}^{K}\sum_{i=1}^{n_F}s^{[i,k]}\Tr\Big(\mathbf{P}^{[i,k]}(\mathbf{P}^{[i,k]})^H\Big) \le P_\mathrm{T},\notag\\
&\mbox{C2:}& \sum_{k=1}^{K}\sum_{i=1}^{n_F}s^{[i,k]}\Tr\Big(\mathbf{G}^{[i,k]}\Big) \le M P_{\max}, \notag\\
&\mbox{C3:}& \sum_{k=1}^K s^{[i,k]}\le 1, \quad\forall i,\notag\\
&\mbox{C4:}&s^{[i,k]}\in\{0,1\}, \
\ \forall i,k,
\end{eqnarray}
where $\Tr\Big(\mathbf{G}^{[i,k]}\Big)$ is the total
power transmitted by  the SUDAS on subcarrier $i$ for UE $k$ and
\begin{eqnarray}\notag
\hspace*{-1mm}\mathbf{G}^{[i,k]}\hspace*{-3mm}&=&\hspace*{-3mm}
\mathbf{F}^{[i,k]}\Big(\mathbf{H}^{[i]}_{\mathrm{B}\rightarrow\mathrm{S}}\mathbf{P}^{[i,k]}
(\mathbf{P}^{[i,k]})^H(\mathbf{H}^{[i]}_{\mathrm{B}\rightarrow\mathrm{S}})^H
\hspace*{-1mm}+\hspace*{-1mm}
\mathbf{I}_M\Big)(\mathbf{F}^{[i,k]})^H\hspace*{-1mm}.
\end{eqnarray}
Constants $P_{\mathrm{T}}$ and $M P_{\max}$ in C1 and C2 are the  maximum transmit power allowances for the BS and the SUDAS ($M$ SUDACs), respectively, where $P_{\max}$ is the average transmit power budget for a SUDAC.  Constraints C3 and C4 are imposed to
guarantee that each subcarrier serves at most one UE.


\subsection{Transformation of the Optimization Problem}
\label{transf_opt_prob}
The considered optimization problem consists of a non-convex objective function and combinatorial constraints which do not facilitate a tractable resource allocation algorithm design. In order to obtain an efficient
resource allocation algorithm, we   study the structure of the optimal precoding policy. For this purpose, we now define the following matrices
before stating an important theorem concerning the  structure of the optimal precoding matrices. Using singular
value decomposition (SVD), the channel matrices $\mathbf{H}^{[i]}_{\mathrm{B}\rightarrow\mathrm{S}}$ and
$\mathbf{H}^{[i,k]}_{\mathrm{S}\rightarrow\mathrm{UE}}$ can be written as
 \begin{eqnarray}\label{eqn:SVD_HSR_HRD}
 \mathbf{H}^{[i]}_{\mathrm{B}\rightarrow\mathrm{S}}&=&\mathbf{U}^{[i]}_{\mathrm{B}\rightarrow\mathrm{S}}
 \mathbf{\Lambda}^{[i]}_{\mathrm{B}\rightarrow\mathrm{S}}(\mathbf{V}^{[i]}_{\mathrm{B}\rightarrow\mathrm{S}})^H \,\mbox{and} \notag\\
 \mathbf{H}^{[i,k]}_{\mathrm{S}\rightarrow\mathrm{UE}}&=&\mathbf{U}^{[i,k]}_{\mathrm{S}\rightarrow\mathrm{UE}}\mathbf{\Lambda}^{[i,k]}_{\mathrm{S}\rightarrow\mathrm{UE}}(\mathbf{V}^{[i,k]}_{\mathrm{S}\rightarrow\mathrm{UE}})^H ,
 \end{eqnarray}
 respectively, where $\mathbf{U}^{[i]}_{\mathrm{B}\rightarrow\mathrm{S}}\in{\mathbb{C}^{ M\times M}},\mathbf{V}^{[i]}_{\mathrm{B}\rightarrow\mathrm{S}}\in{\mathbb{C}^{N_{\mathrm{T}}\times N_{\mathrm{T}}}},
 \mathbf{U}^{[i,k]}_{\mathrm{S}\rightarrow\mathrm{UE}}\in{\mathbb{C}^{M\times M}},$ and  $\mathbf{V}^{[i,k]}_{\mathrm{S}\rightarrow\mathrm{UE}}\in{\mathbb{C}^{M\times M}}$
 are unitary matrices. $\mathbf{\Lambda}^{[i]}_{\mathrm{B}\rightarrow\mathrm{S}}$ and $\mathbf{\Lambda}^{[i,k]}_{\mathrm{S}\rightarrow\mathrm{UE}}$ and  are $ M \times  N_{\mathrm{T}}$ and  $ M \times M$ matrices with main diagonal
 element vectors $\diag(\mathbf{\Lambda}^{[i]}_{\mathrm{B}\rightarrow\mathrm{S}})=\Big[\sqrt{\gamma_{\mathrm{B}\rightarrow\mathrm{S},1}^{[i]}}\, \sqrt{\gamma_{\mathrm{B}\rightarrow\mathrm{S},2}^{[i]}}\,\ldots\sqrt{\gamma_{\mathrm{B}\rightarrow\mathrm{S},R_1}^{[i]}}\Big]$
  and  $\diag(\mathbf{\Lambda}^{[i,k]}_{\mathrm{S}\rightarrow\mathrm{UE}})= \Big[\sqrt{\gamma_{\mathrm{S}\rightarrow\mathrm{UE},1}^{[i,k]}}\,
  \sqrt{\gamma_{\mathrm{S}\rightarrow\mathrm{UE},2}^{[i,k]}}\,$ $\ldots\,
  \sqrt{\gamma_{\mathrm{S}\rightarrow\mathrm{UE},R_2}^{[i,k]}}\Big]$, where the elements are ordered in ascending order, respectively. $R_1=\Rank(\mathbf{H}^{[i]}_{\mathrm{B}\rightarrow\mathrm{S}})$ and $R_2=\Rank(\mathbf{H}^{[i,k]}_{\mathrm{S}\rightarrow\mathrm{UE}})$.   Variables $\gamma_{\mathrm{B}\rightarrow\mathrm{S},n}^{[i]}$ and $\gamma_{\mathrm{S}\rightarrow\mathrm{UE},n}^{[i,k]}$ represent the equivalent channel-to-noise ratio (CNR) on spatial
   channel $n$ in subcarrier $i$ of the BS-to-SUDAS channel and the SUDAS-to-UE $k$ channel, respectively.

 We are now ready to introduce the following theorem.
\begin{Thm}\label{Thm:Diagonalization_optimal}
Suppose that
$\Rank(\mathbf{P}^{[i,k]})=\Rank(\mathbf{F}^{[i,k]})=N_{\mathrm{S}}\le
 \min\{\Rank(\mathbf{H}^{[i]}_{\mathrm{S}\rightarrow\mathrm{UE}}),\Rank(\mathbf{H}^{[i]}_{\mathrm{B}\rightarrow\mathrm{S}})\}$.
 In this case, the optimal linear
  precoding matrices used at the BS and the SUDACs jointly diagonalize the BS-to-SUDAS-to-UE channels on each subcarrier, despite the non-convexity of the objective function.  The optimal precoding matrices are given by
\begin{eqnarray}\label{eqn:matrix_P}
\mathbf{P}^{[i,k]}&=&\mathbf{\widetilde V}^{[i]}_{\mathrm{B}\rightarrow\mathrm{S}}\mathbf{\Lambda}^{[i,k]}_{\mathrm{B}}
\ \mbox{and}\ \notag\\
\mathbf{F}^{[i,k]}&=&\mathbf{\widetilde V}^{[i,k]}_{\mathrm{S}\rightarrow\mathrm{UE}}\mathbf{\Lambda}^{[i,k]}_{\mathrm{F}}(\mathbf{\widetilde U}^{[i,k]}_{\mathrm{B}\rightarrow\mathrm{S}})^H
      \end{eqnarray}
respectively, where $\mathbf{\widetilde V}^{[i]}_{\mathrm{B}\rightarrow\mathrm{S}}$, $\mathbf{\widetilde V}^{[i,k]}_{\mathrm{S}\rightarrow\mathrm{UE}}$, and $\mathbf{\widetilde U}^{[i,k]}_{\mathrm{B}\rightarrow\mathrm{S}}$ are the $N_{\mathrm{S}}$ rightmost columns of $\mathbf{ V}^{[i]}_{\mathrm{B}\rightarrow\mathrm{S}}$, $\mathbf{ V}^{[i,k]}_{\mathrm{S}\rightarrow\mathrm{UE}}$, and $\mathbf{ U}^{[i,k]}_{\mathrm{B}\rightarrow\mathrm{S}}$, respectively.   Matrices $\mathbf{\Lambda}^{[i,k]}_{\mathrm{B}}\in
\mathbb{C}^{N_{\mathrm{S}}\times N_{\mathrm{S}}}$ and $\mathbf{\Lambda}^{[i,k]}_{\mathrm{F}}\in
\mathbb{C}^{N_{\mathrm{S}} \times N_{\mathrm{S}}}$ are diagonal matrices with diagonal element
vectors
$\diag(\mathbf{\Lambda}^{[i,k]}_{\mathrm{B}})=\Big[\sqrt{P_{\mathrm{B}\rightarrow\mathrm{S},1}^{[i,k]}}
\,\ldots\,\sqrt{P_{\mathrm{B}\rightarrow\mathrm{S},n}^{[i,k]}}\,\ldots\,\sqrt{P_{\mathrm{B}\rightarrow\mathrm{S},N_{\mathrm{S}}}^{[i,k]}}\Big]$,
 and
$\diag(\mathbf{\Lambda}^{[i,k]}_{\mathrm{F}})=\Big[\sqrt{P_{\mathrm{S}\rightarrow\mathrm{UE},1}^{[i,k]}}
 \,\ldots\,\sqrt{P_{\mathrm{S}\rightarrow\mathrm{UE},n}^{[i,k]}}\,\ldots\,
\sqrt{P_{\mathrm{S}\rightarrow\mathrm{UE},N_{\mathrm{S}}}^{[i,k]}}\Big]$, respectively. Variables $P_{\mathrm{B}\rightarrow\mathrm{S},n}^{[i,k]}$
and $P_{\mathrm{S}\rightarrow\mathrm{UE},n}^{[i,k]}$ are, respectively, the equivalent
transmit powers of the BS-to-SUDAS link and the SUDAS-to-UE link for UE $k$ on spatial channel $n$ in subcarrier $i$.
\end{Thm}

\begin{proof}
Please refer to the Appendix.
\end{proof}

By Theorem \ref{Thm:Diagonalization_optimal}, the end-to-end MIMO
channel on subcarrier $i$ is converted into $N_{\mathrm{S}}$ parallel spatial
channels if the optimal precoding matrices are
used.

Therefore, the achievable rate on subcarrier $i$ between the BS and UE $k$ via the SUDAS  can
be expressed as \cite{JR:kwan_AF_relay,JR:MIMO_HD_relay1}:
\begin{eqnarray}
 R^{[i,k]}&=&\sum_{n=1}^{N_{\mathrm{S}}}\log_2(1+\mathrm{SINR}^{[i,k]}_n)\quad \mbox{where}\\
\mathrm{SINR}^{[i,k]}_n&=&
\notag\frac{\gamma_{\mathrm{B}\rightarrow\mathrm{S},n}^{[i]}P_{\mathrm{B}\rightarrow\mathrm{S},n}^{[i,k]}
P_{\mathrm{S}\rightarrow\mathrm{UE},n}^{[i,k]}\gamma_{\mathrm{S}\rightarrow\mathrm{UE},n}^{[i,k]}}{1+\gamma_{\mathrm{B}\rightarrow\mathrm{S},n}^{[i]}P_{\mathrm{B}\rightarrow\mathrm{S},n}^{[i,k]}
+ P_{\mathrm{S}\rightarrow\mathrm{UE},n}^{[i,k]}\gamma_{\mathrm{S}\rightarrow\mathrm{UE},n}^{[i,k]}}
\end{eqnarray}
is the received signal-to-interference-plus-noise-ratio (SINR). Although the objective function is now a scalar function with respect  to the optimization variables, it is still non-convex.  To obtain a tractable resource allocation algorithm design, we propose the following objective function approximation. In particular, the end-to-end SINR on subcarrier $i$ for UE $k$ can be approximated as
\begin{eqnarray}\label{eqn:SINR_1}
\mathrm{SINR}^{[i,k]}_n\hspace*{-1mm} &\approx&\hspace*{-1mm}\widetilde{\mathrm{SINR}}^{[i,k]}_n\quad \mbox{where}\\
\widetilde{\mathrm{SINR}}^{[i,k]}_n\hspace*{-1mm}&=&\hspace*{-1mm}\frac{\gamma_{\mathrm{B}\rightarrow\mathrm{S},n}^{[i]}P_{\mathrm{B}\rightarrow\mathrm{S},n}^{[i,k]}
P_{\mathrm{S}\rightarrow\mathrm{UE},n}^{[i,k]}\gamma_{\mathrm{S}\rightarrow\mathrm{UE},n}^{[i,k]}}{\gamma_{\mathrm{B}\rightarrow\mathrm{S},n}^{[i]}P_{\mathrm{B}\rightarrow\mathrm{S},n}^{[i,k]}
+ P_{\mathrm{S}\rightarrow\mathrm{UE},n}^{[i,k]}\gamma_{\mathrm{S}\rightarrow\mathrm{UE},n}^{[i,k]}}\label{eqn:SINR_2}.
\end{eqnarray}
\newcounter{mytempeqncnt}
\begin{figure*}[ht]
\setcounter{mytempeqncnt}{\value{equation}}
\setcounter{equation}{16}
 \begin{eqnarray}\label{eqn:P1}
\hspace*{-4mm}P_{\mathrm{B}\rightarrow\mathrm{S},n}^{[i,k]}\hspace*{-3mm}&=&\hspace*{-2mm}\Bigg[
   \frac{\gamma_{\mathrm{S}\rightarrow\mathrm{UE},n}^{[i,k]} P_{\mathrm{S}\rightarrow\mathrm{UE},n}^{[i,k]} \left(\hspace*{-0.5mm}\frac{\sqrt{4 w^{[k]} \gamma_{\mathrm{B}\rightarrow\mathrm{S},n}^{[i]}(1+
   \gamma_{\mathrm{S}\rightarrow\mathrm{UE},n}^{[i,k]} P_{\mathrm{S}\rightarrow\mathrm{UE},n}^{[i,k]})\hspace*{-0.5mm}+\hspace*{-0.5mm}(\gamma_{\mathrm{S}\rightarrow\mathrm{UE},n}^{[i,k]})^2 \lambda
   (P_{\mathrm{S}\rightarrow\mathrm{UE},n}^{[i,k]})^2 \ln(2)}}{\sqrt{\lambda } \sqrt{\ln
   (2)}}\hspace*{-0.5mm}-\hspace*{-0.5mm}\gamma_{\mathrm{S}\rightarrow\mathrm{UE},n}^{[i,k]}P_{\mathrm{S}\rightarrow\mathrm{UE},n}^{[i,k]}\hspace*{-0.5mm}-\hspace*{-0.5mm}2\right)}{2 (\gamma_{\mathrm{B}\rightarrow\mathrm{S},n}^{[i]} \gamma_{\mathrm{S}\rightarrow\mathrm{UE},n}^{[i,k]}
   P_{\mathrm{S}\rightarrow\mathrm{UE},n}^{[i,k]}\hspace*{-0.5mm}+\hspace*{-0.5mm}\gamma_{\mathrm{B}\rightarrow\mathrm{S},n}^{[i]})}\Bigg]^+\\
 \label{eqn:P2}\hspace*{-4mm}  P_{\mathrm{S}\rightarrow\mathrm{UE},n}^{[i,k]}\hspace*{-3mm}&=& \hspace*{-2mm}\Bigg[\frac{\gamma_{\mathrm{B}\rightarrow\mathrm{S},n}^{[i]} P_{\mathrm{B}\rightarrow\mathrm{S},n}^{[i,k]} \left(\frac{\sqrt{(\gamma_{\mathrm{B}\rightarrow\mathrm{S},n}^{[i]})^2 \beta (P_{\mathrm{B}\rightarrow\mathrm{S},n}^{[i,k]})^2 \ln(2)
 + (\gamma_{\mathrm{B}\rightarrow\mathrm{S},n}^{[i]} P_{\mathrm{B}\rightarrow\mathrm{S},n}^{[i,k]}+1) 4w^{[k]} \gamma_{\mathrm{S}\rightarrow\mathrm{UE},n}^{[i,k]}}}{\sqrt{\beta }
   \sqrt{\ln (2)}}-\gamma_{\mathrm{B}\rightarrow\mathrm{S},n}^{[i]} P_{\mathrm{B}\rightarrow\mathrm{S},n}^{[i,k]}-2\right)}{2 (\gamma_{\mathrm{B}\rightarrow\mathrm{S},n}^{[i]} \gamma_{\mathrm{S}\rightarrow\mathrm{UE},n}^{[i,k]} P_{\mathrm{B}\rightarrow\mathrm{S},n}^{[i,k]}+\gamma_{\mathrm{S}\rightarrow\mathrm{UE},n}^{[i,k]})}\Bigg]^+
   \end{eqnarray}
  \hrulefill
\end{figure*}We note that this approximation is asymptotically tight in high SNR. The next step is to handle the combinatorial constraint C4 in (\ref{eqn:cross-layer-formulation}). To this end, we adopt the time-sharing relaxation approach. In particular, we relax $ s^{[i,k]}$ in C4 such that it is a real valued optimization variable between zero and one \cite{JR:Roger_OFDMA,JR:Time_sharing_wei_yu,JR:Kwan_DF}, i.e., $0\le  s^{[i,k]}\le 1$. It is shown in \cite{JR:Time_sharing_wei_yu} that the time-sharing relaxation is asymptotically optimal for a sufficient number of subcarriers\footnote{It has been shown by simulation in
\cite{CN:large_subcarriers} that the performance gap due to time-sharing relaxation  is virtually zero even for OFDMA systems with only   $8$ subcarriers.}. Next, we define two auxiliary optimization variables $\tilde P_{\mathrm{B}\rightarrow\mathrm{S},n}^{[i,k]}=s^{[i,k]}P_{\mathrm{B}\rightarrow\mathrm{S},n}^{[i,k]}$ and $\tilde P_{\mathrm{S}\rightarrow\mathrm{UE},n}^{[i,k]}=s^{[i,k]}P_{\mathrm{S}\rightarrow\mathrm{UE},n}^{[i,k]}$ and rewrite the optimization problem as:\setcounter{equation}{15}
\begin{eqnarray}\label{eqn:cross-layer-formulation-transformed}
&&\hspace*{-10mm}\underset{{{\widetilde P_{\mathrm{B}\rightarrow\mathrm{S}}^{[i,k]},\widetilde P_{\mathrm{S}\rightarrow\mathrm{UE},n}^{[i,k]}},}{\cal S}}\maxo\ \
\sum_{k=1}^K\sum_{i=1}^{N_{\mathrm{F}}}\sum_{n=1}^{N_{\mathrm{S}}}s^{[i,k]}\log_2\Big(1+\frac{\overline{\mathrm{{SINR}}}^{[i,k]}_n}{s^{[i,k]}}\Big) \notag\\
\mbox{s.t.} &\mbox{C1:}& \sum_{k=1}^{K}\sum_{i=1}^{n_F}\sum_{n=1}^{N_{\mathrm{S}}}\widetilde P_{\mathrm{B}\rightarrow\mathrm{S},n}^{[i,k]} \le P_\mathrm{T},\notag\\
&\mbox{C2:}& \sum_{k=1}^{K}\sum_{i=1}^{n_F}\sum_{n=1}^{N_{\mathrm{S}}} \widetilde  P_{\mathrm{S}\rightarrow\mathrm{UE},n}^{[i,k]}\le M P_{\max}, \notag\\
&\mbox{C3:}& \sum_{k=1}^K s^{[i,k]}\le 1, \quad\forall i,\quad \mbox{C4: } 0\le s^{[i,k]}\le 1, \
\ \forall i,k,\notag\\
& \mbox{C5:}&
\widetilde  P_{\mathrm{B}\rightarrow\mathrm{S},n}^{[i,k]}, \widetilde  P_{\mathrm{S}\rightarrow\mathrm{UE},n}^{[i,k]}\ge 0, \quad
\forall i,k,n,\quad
\end{eqnarray}\addtocounter{equation}{2}
where $\overline{\mathrm{{SINR}}}^{[i,k]}_n=\widetilde{\mathrm{SINR}}^{[i,k]}_n\Big|_{\Phi}$ and $\Phi=\{P_{\mathrm{B}\rightarrow\mathrm{S},n}^{[i,k]}=\tilde P_{\mathrm{B}\rightarrow\mathrm{S},n}^{[i,k]}/s^{[i,k]},P_{\mathrm{S}\rightarrow\mathrm{UE},n}^{[i,k]}=\tilde P_{\mathrm{S}\rightarrow\mathrm{UE},n}^{[i,k]}/s^{[i,k]}\}$. It can be shown that optimization problem (\ref{eqn:cross-layer-formulation-transformed}) is jointly concave with respect to the auxiliary optimization variables and $s^{[i,k]}$. We note that by solving  optimization problem (\ref{eqn:cross-layer-formulation-transformed})  for $\widetilde P_{\mathrm{B}\rightarrow\mathrm{S}}^{[i,k]}$, $\widetilde P_{\mathrm{S}\rightarrow\mathrm{UE},n}^{[i,k]}$, and $s^{[i,k]}$, we can
recover the solution for $ P_{\mathrm{B}\rightarrow\mathrm{S}}^{[i,k]}$ and $P_{\mathrm{S}\rightarrow\mathrm{UE},n}^{[i,k]}$.  In other words, the solution of (\ref{eqn:cross-layer-formulation-transformed}) is asymptotically optimal with respect to (\ref{eqn:cross-layer-formulation}) for high SNR and a sufficiently large number of subcarriers.

In the following, we propose an algorithm for solving the transformed problem in (\ref{eqn:cross-layer-formulation-transformed}).

\begin{table}[t]\caption{Iterative Resource Allocation Algorithm}\label{table:algorithm}
\vspace*{-0.6cm}
\renewcommand\thealgorithm{}
\begin{algorithm} [H]                    
\caption{Alternating Optimization}          
\label{alg1}                           
\begin{algorithmic} [1]
\STATE Initialize the maximum number of iterations $L_{\max}$ and a small constant $\kappa\rightarrow 0$
\STATE Set iteration index $l=0$ and initialize a feasible solution point $P_{\mathrm{B}\rightarrow\mathrm{S},n}^{[i,k]}(l),P_{\mathrm{S}\rightarrow\mathrm{UE},n}^{[i,k]}(l),s^{[i,k]}(l)$, $l=l+1$

\REPEAT [Loop]
\STATE Solve  (\ref{eqn:cross-layer-formulation-transformed})  for $P_{\mathrm{B}\rightarrow\mathrm{S},n}^{[i,k]}$ and $s^{[i,k]}$ with a fixed $P_{\mathrm{S}\rightarrow\mathrm{UE},n}^{[i,k]}(l-1)$ by using (\ref{eqn:P1}) and (\ref{eqn:sub_selection}) which leads to  intermediate power allocation variables $P_{\mathrm{B}\rightarrow\mathrm{S},n}^{[i,k]'}$ and a subcarrier allocation policy $s^{[i,k]'}$

\STATE Solve  (\ref{eqn:cross-layer-formulation-transformed})  for $P_{\mathrm{S}\rightarrow\mathrm{UE},n}^{[i,k]}$ and $s^{[i,k]}$ with $P_{\mathrm{B}\rightarrow\mathrm{S},n}^{[i,k]'}$ via equation (\ref{eqn:P2}) and (\ref{eqn:sub_selection}) which leads to  intermediate power allocation variables $P_{\mathrm{S}\rightarrow\mathrm{UE},n}^{[i,k]'}$ and a subcarrier allocation policy $s^{[i,k]''}$
\IF{$\abs{P_{\mathrm{S}\rightarrow\mathrm{UE},n}^{[i,k]'}-P_{\mathrm{S}\rightarrow\mathrm{UE},n}^{[i,k]}(l-1)}\le \kappa$ and $\abs{P_{\mathrm{S}\rightarrow\mathrm{UE},n}^{[i,k]'}-P_{\mathrm{S}\rightarrow\mathrm{UE},n}^{[i,k]}(l-1)}\le \kappa$ and $\abs{s^{[i,k]'}-s^{[i,k]}(l-1)}\le \kappa$ }
\STATE
Convergence = \TRUE \RETURN
$\{P_{\mathrm{S}\rightarrow\mathrm{UE},n}^{[i,k]'},P_{\mathrm{B}\rightarrow\mathrm{S},n}^{[i,k]'},s^{[i,k]''}\}$
\ELSE
\STATE $P_{\mathrm{B}\rightarrow\mathrm{S},n}^{[i,k]}(l)=P_{\mathrm{B}\rightarrow\mathrm{S},n}^{[i,k]'},
P_{\mathrm{S}\rightarrow\mathrm{UE},n}^{[i,k]}(l)=P_{\mathrm{S}\rightarrow\mathrm{UE},n}^{[i,k]'},s^{[i,k]}(l)=s^{[i,k]''}$, $l=l+1$
 \ENDIF
\UNTIL{ $l=L_{\max}$}

\end{algorithmic}
\end{algorithm}
\normalsize
\end{table}

\subsection{Iterative Resource Allocation Algorithm}
The proposed iterative resource allocation algorithm is based on alternating  optimization. The algorithm is summarized in Table \ref{table:algorithm}. The algorithm is implemented by a repeated loop.  In line 2,  we first set the iteration index $l$ to zero
and initialize the resource allocation policy. Variables $P_{\mathrm{B}\rightarrow\mathrm{S},n}^{[i,k]}(l),P_{\mathrm{S}\rightarrow\mathrm{UE},n}^{[i,k]}(l)$ and $s^{[i,k]}(l)$ denote the resource allocation policy in the $l$-th iteration. Then, in each iteration, we solve  (\ref{eqn:cross-layer-formulation-transformed})  for $P_{\mathrm{B}\rightarrow\mathrm{S},n}^{[i,k]}$ using (\ref{eqn:P1}) with $s^{[i,k]},\forall i,k,$ and  $P_{\mathrm{S}\rightarrow\mathrm{UE},n}^{[i,k]}(l-1)$ from the last iteration. Then, we obtain an intermediate power allocation variable $P_{\mathrm{B}\rightarrow\mathrm{S},n}^{[i,k]'}$ which is used as an input for solving (\ref{eqn:cross-layer-formulation-transformed}) for $P_{\mathrm{S}\rightarrow\mathrm{UE},n}^{[i,k]}$ via (\ref{eqn:P2}), c.f., line 5.  We note that (\ref{eqn:P1}) and (\ref{eqn:P2}) are obtained by standard convex optimization techniques while variables $\lambda$ and $\beta$  in (\ref{eqn:P1}) and (\ref{eqn:P2}) are the Lagrange multipliers  with  respect to constraints C1 and C2 in (\ref{eqn:cross-layer-formulation-transformed}), respectively. The optimal values of $\lambda$ and $\beta$ in each iteration can be easily found with a standard gradient algorithm such that constraints C1 and C2  in (\ref{eqn:cross-layer-formulation-transformed}) are satisfied. After we obtain the intermediate solution for power allocation, we use it to derive   the optimal allocation of subcarrier $i$ at the
BS to UE $k$ which is given by
\begin{eqnarray}
\label{eqn:sub_selection}s^{[i,k]*}=
 \left\{ \begin{array}{rl}
 1 &\mbox{if $k =\arg\underset{t\in\{1,\ldots K\}}\max \,\Psi_t$, }  \\
 0 &\mbox{ otherwise,}
       \end{array} \right.
\end{eqnarray}
and
\begin{eqnarray}\label{eqn:subcarrier_allocation}
\hspace*{-5.5mm}\Psi_k\hspace*{-2mm}&=&\hspace*{-3mm}w^{[k]}\Big(\hspace*{-0.5mm}\sum_{n=1}^N\hspace*{-0.5mm}\log_2\Big(1+\widetilde{\mathrm{SINR}}^{[i,k]*}_n\Big)
\hspace*{-0.5mm}-\hspace*{-0.5mm}
\frac{\widetilde{\mathrm{SINR}}^{[i,k]*}_n}{1+\widetilde{\mathrm{SINR}}^{[i,k]*}_n}\hspace*{-0.5mm}\Big).
\end{eqnarray}
$\widetilde{\mathrm{SINR}}^{[i,k]*}_n$ is obtained by substituting the intermediate solution of $P_{\mathrm{B}\rightarrow\mathrm{S},n}^{[i,k]'}$ and $P_{\mathrm{S}\rightarrow\mathrm{UE},n}^{[i,k]'}$, i.e., (\ref{eqn:P1}) and (\ref{eqn:P2}), into (\ref{eqn:SINR_2}) in the $l$-th iteration. Then, the  procedure is repeated iteratively until we reach the maximum number of iterations or convergence is achieved.  We note that the convergence to the optimal solution of (\ref{eqn:cross-layer-formulation-transformed}) is guaranteed for a large number of iterations since  (\ref{eqn:cross-layer-formulation-transformed}) is concave with respect to the optimization variables \cite{JR:AO}. Besides, the proposed algorithm has a polynomial time computational complexity.

\section{Results and Discussions}\label{sect:result-discussion}
In this section, we evaluate the system performance based on simulations. We assume that there are $K$ UEs located in an indoor environment and the BS is located outdoor.
 For the BS-to-SUDAS links, we adopt the Urban macro outdoor-to-indoor scenario of the Wireless World Initiative New Radio (WINNER+) channel model \cite{Spec:Winner+}. The center frequency  and the bandwidth of the licensed band are  $800$ MHz and $10$ MHz, respectively.  There are $600$  subcarriers which are grouped into $50$ resource blocks with $12$ subcarriers per resource block for data transmission. Each subcarrier has a bandwidth of $15$ kHz. Hence, the BS-to-SUDAS link configuration is in accordance with the Long Term Evolutions (LTE) standard \cite{Spec:LTE}.  As for the SUDAS-to-UE links, we adopt the IEEE $802.11$ad channel model \cite{Spec:60GHz} in the range of $60$ GHz. There are $M$ subbands. The first subband has a frequency range from $60$ GHz to $60.01$ GHz and there is a $30$ MHz  guard band between any two consecutive subbands. The maximum transmit power per SUDAC is set to $P_{\max}=23$ dBm which is in accordance with the maximum power spectral density suggested by the Harmonized European Standard  \cite{Spec:BRAN}, i.e.,  $13$ dBm-per-MHz. For simplicity, we assume that  $w^{[k]}=1,\forall k$, and $N_{\mathrm{S}}=\min\{N_{\mathrm{T}},M\}$ for
studying the system performance.
\begin{figure}[t]
 \centering
\includegraphics[width=3.5in]{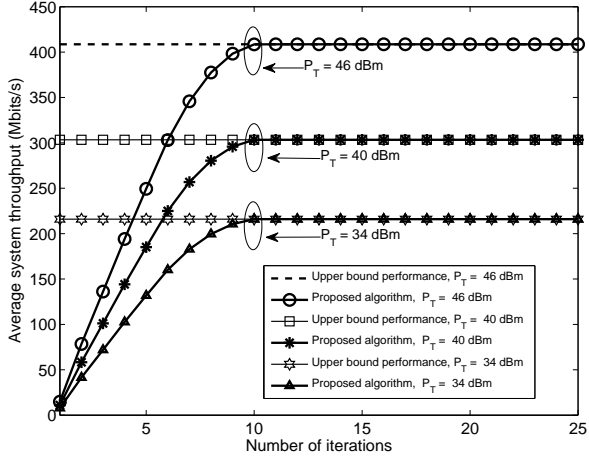}\vspace*{-1mm}
 \caption{Average system throughput (Mbits/s) versus the number of iterations for $N_{\mathrm{T}}=8$ transmit antennas at the BS, $K=2$ UEs, $M=8$ SUDACs, and different maximum transmit powers at the BS, $P_{\mathrm{T}}$.}
 \label{fig:convergence}
\end{figure}

\subsection{Convergence of the Proposed Iterative Algorithm}
Figure \ref{fig:convergence} illustrates the convergence of
the proposed algorithm  for $N_{\mathrm{T}}=8$ transmit antennas at the BS, $K=2$ UEs, $M=8$ SUDACs, and different maximum transmit powers at the BS, $P_{\mathrm{T}}$. We compare the system performance of the proposed algorithm with a performance upper bound  which is obtained by computing the optimal objective value in (\ref{eqn:cross-layer-formulation-transformed}), i.e., assuming noise free reception at the UEs. The performance gap between the two curves constitutes an upper bound on the performance loss due to the high SINR approximation adopted in (\ref{eqn:SINR_1}). It can be seen that
the proposed algorithm approaches $99\%$ of the upper bound
value after $20$ iterations which  confirms the practicality of the proposed iterative algorithm.

\subsection{Average System Throughput versus Transmit Power}
Figure \ref{fig:TP_PT} illustrates the average system
throughput versus the maximum transmit power at the BS for $K=2$ UEs, $M=8$, and different numbers of transmit antennas $N_{\mathrm{T}}$ at the BS. The
performance of the proposed algorithm for the SUDAS with $10$ iterations
is compared with that of a benchmark
MIMO system and a baseline system.
For the benchmark MIMO system, we assume that each UE is equipped with $M$ receive antennas without the help of the SUDAS and optimal resource allocation is performed to maximize the  system throughput. In other words, the average system throughput of the benchmark system serves as a performance upper bound for the proposed SUDAS.  As for the baseline system, the BS performs optimal resource allocation and utilizes only the licensed frequency band without the help of the SUDAS and the UEs have only one antenna.  As can be observed, the proposed SUDAS is able to exploit the spatial multiplexing gain even though each UE is equipped with a single antenna. Besides, a huge performance gain is achieved by the SUDAS compared   to the baseline system as the SUDAS utilizes both licensed and unlicensed bands. On the other hand, the performance of the proposed scheme and the benchmark system improves rapidly with  increasing number of transmit antennas  due to  more degrees of freedom for resource allocation.
\begin{figure}[t]
 \centering
\includegraphics[width=3.5in]{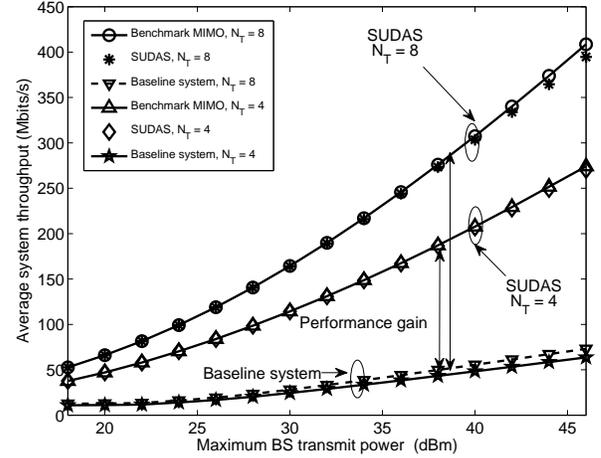}\vspace*{-1mm}
 \caption{Average system throughput (Mbits/s) versus the maximum transmit power at the BS (dBm)  for $K=2$ UEs, $M=8$ SUDACs, and different communication systems. The double-sided arrows indicate the throughput gains achieved by the SUDAS compared to the baseline system.}
 \label{fig:TP_PT}
\end{figure}
\begin{figure}[t]
 \centering
\includegraphics[width=3.5in]{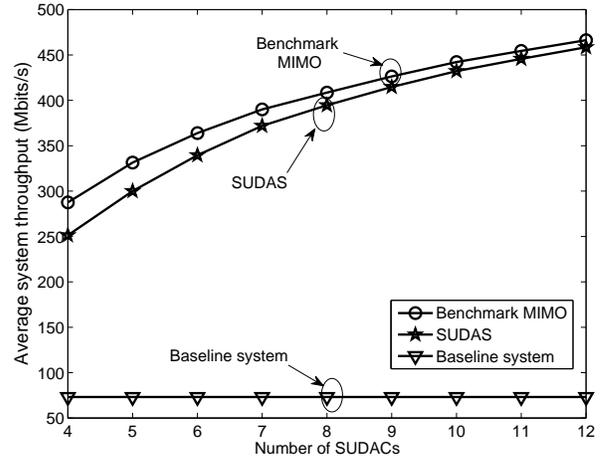}\vspace*{-1mm}
 \caption{Average system throughput (Mbits/s) versus the number of SUDACs for a maximum transmit power at the BS of $P_{\mathrm{T}}=46$ dBm for  different communication systems.}
 \label{fig:TP_UE}
\end{figure}
\subsection{Average System Throughput versus Number of SUDACs}
Figure \ref{fig:TP_UE} depicts the average system throughput
versus the number of SUDACs for $N_{\mathrm{T}}=8$.
The maximum BS transmit power is $46$ dBm.  It can be observed that the
system throughput grows with the number of SUDACs. In particular, a higher spatial multiplexing gain can be achieved when we increase the number of SUDACs $M$ if $N_{\mathrm{T}}\ge M$. For $M>N_{\mathrm{T}}$, increasing the number of SUDACs in the system leads to more  spatial diversity which also improves the system throughput.

\section{Conclusion}\label{sect:conclusion}
In this paper, we studied the resource allocation algorithm design for SUDAS assisted downlink  multicarrier transmission.
In particular, the SUDAS utilizes both licensed and unlicensed frequency bands  for improving the system throughput.  The resource allocation algorithm design was formulated as a non-convex matrix optimization problem. In order to obtain a tractable solution,
we  revealed the structures of the optimal precoding  matrices such that the  problem could be
transformed into a scalar optimization problem. Based on this result, we proposed an efficient
iterative  resource allocation algorithm to solve the problem by
alternating optimization. Our simulation results show that the proposed
SUDAS assisted transmission provides a substantial throughput gain compared to conventional systems which do not utilize the unlicensed frequency spectrum. It is expected that the proposed SUDAS is able to bridge the gap between the current technology and the high data rate requirement of the next generation communication systems.

\section*{Appendix-Proof of Theorem 1}
Due to the page limitation, we provide only a sketch of the proof  which follows a similar proof in \cite{JR:Yue_Rong_diagonalization,JR:Kwan_FD}.  We first show that the optimal
$\mathbf{P}^{[i,k]}$ and $\mathbf{F}^{[i,k]}$ jointly diagonalize the end-to-end channels on
each subcarrier for the maximization of the objective function in (\ref{eqn:cross-layer-formulation}). Then, we construct
the optimal precoding and post-processing matrices based on the optimal structure. The MSE matrix for data transmission on subcarrier $i$ for UE $k$ can be written as:
\begin{eqnarray}
\mathbf{E}^{[i,k]}&\hspace*{-4mm}=&\hspace*{-3mm}\Big(\mathbf{I}_{N_\mathrm{S}}+(\mathbf{\Gamma}^{[i,k]})^H(\mathbf{\Theta}^{[i]})^{-1}\mathbf{\Gamma}^{[i,k]}\Big)^{-1}
\\
&\hspace*{-4mm}=&\hspace*{-3mm}\notag\mathbf{I}_{N_\mathrm{S}}-(\mathbf{\Gamma}^{[i,k]})^H\Big(\mathbf{\Gamma}^{[i,k]}(\mathbf{\Gamma}^{[i,k]})^H+\mathbf{\Theta}^{[i]}\Big)^{-1}\mathbf{\Gamma}^{[i,k]}.
\end{eqnarray}
Since the objective function for each subcarrier is a Schur-concave function, by applying the majorization theory \cite{book:majorization}, it can be shown that   the sum of the diagonal elements of the MSE matrix is minimized when  matrix $(\mathbf{\Gamma}^{[i,k]})^H\Big(\mathbf{\Gamma}^{[i,k]}(\mathbf{\Gamma}^{[i,k]})^H+\mathbf{\Theta}^{[i]}\Big)^{-1}\mathbf{\Gamma}^{[i,k]}$ is a diagonal matrix. In other words, the objective function is maximized when the MSE matrix is a diagonal matrix.

On the other hand, we focus on the power consumption constraints C1 and C2 in (\ref{eqn:cross-layer-formulation}). For a given subcarrier allocation and a given achievable data rate, it can be shown that
the transmit powers at the BS and the SUDAS are minimized when
 matrices
$\mathbf{P}^{[i,k]}$ and $\mathbf{F}^{[i,k]}$ are given by
\begin{eqnarray}\label{eqn:optimal_B}
\mathbf{P}^{[i,k]}&=&\mathbf{\widetilde V}^{[i]}_{\mathrm{B}\rightarrow\mathrm{S}}\mathbf{\Lambda}^{[i,k]}_{\mathrm{B},k}\quad \mbox{and}\\
\label{eqn:optimal_F}\mathbf{F}^{[i,k]}
&=&\mathbf{\widetilde V}^{[i]}_{R_m,k}\mathbf{\Lambda}^{[i,k]}_{\mathrm{F},k}(\mathbf{\widetilde U}^{[i]}_{\mathrm{B}\rightarrow\mathrm{S}})^H,
\end{eqnarray}
respectively, where all involved matrices are defined after (\ref{eqn:matrix_P}).  Since both $\mathbf{P}^{[i,k]}$ and $\mathbf{F}^{[i,k]}$
 in (\ref{eqn:optimal_B}) and (\ref{eqn:optimal_F}) jointly
diagonalize the end-to-end equivalent  channel and achieve the minimum transmit power for any achievable system data rate, (\ref{eqn:optimal_B}) and (\ref{eqn:optimal_F}) are the optimal precoding and post-processing matrices.

\qed

\bibliographystyle{IEEEtran}
\bibliography{IEEEabrv,OFDMA-AF}

\end{document}